\documentclass[12pt]{amsart}

\usepackage{amsmath,amsthm,amssymb,amsfonts,verbatim}
\usepackage[colorlinks,
            linkcolor=blue,
            anchorcolor=blue,
            citecolor=blue
            ]{hyperref}
\usepackage[hmargin=1.2in,vmargin=1.2in]{geometry}

\title[Asymptotic LRCs with multiple recovering sets]{Asymptotic construction of locally repairable codes with multiple recovering sets}
\author{Songsong Li} \address{School of Electronics, Information and Electric Engineering, Shanghai Jiao Tong University,
China 200240}\email{songsli@sjtu.edu.cn}
\author{Shu Liu}\address{Natl Key Lab Sci and Technol Commun, University of Electronic Science and Technology of China, Chengdu, China.} \email{shuliu@uestc.edu.cn}
\author{Liming Ma}\address{School of Mathematical Sciences, University of Science and Technology of China, Hefei 230026, China
}\email{lmma20@ustc.edu.cn}
\author{Chaoping Xing} \address{School of Electronics, Information and Electric Engineering, Shanghai Jiao Tong University,
China 200240}\email{xingcp@sjtu.edu.cn}
\date{}

\newtheorem{lemma}{Lemma}[section]
\newtheorem{theorem}[lemma]{Theorem}

\newtheorem{prop}[lemma]{Proposition}

\newtheorem{defn}{Definition}

\theoremstyle{remark}

\renewcommand{\epsilon}{\varepsilon}
\renewcommand{\le}{\leqslant}
\renewcommand{\ge}{\geqslant}

\def\Gal{{\rm Gal}}

\newcommand{\vnote}[1]{}

\def\F{\mathbb{F}}

\def \mC {\mathcal{C}}
\def \mA {\mathcal{A}}

\def \mA {\mathcal{A}}
\def \mB {\mathcal{B}}
\def \mC {\mathcal{C}}

\def \mL {\mathcal{L}}

\def \mP {\mathcal{P}}

\def \Xi {{X^{[i]}}}

\newcommand{\Ga}{\alpha}
\newcommand{\Gb}{\beta}

\newcommand{\Gs}{\sigma}

\def \bc {{\bf c}}

\def\supp {{\rm supp }}

\def\Aut {{\rm Aut }}

\def\LRC {{\rm locally repairable code\ }}
\def\LRCs {{\rm locally repairable codes\ }}

\def\Gal{{\rm Gal}}

\begin{document}
\maketitle

\begin{abstract}
Locally repairable codes  have been extensively investigated due to practical applications in distributed and cloud storage systems in recent years. However, not much work on asymptotic behavior of locally repairable codes has been done. In particular, there is few result on constructive lower bound of asymptotic behavior of locally repairable codes with multiple recovering sets. In this paper, we construct some families of  asymptotically good locally repairable codes with multiple recovering sets  via automorphism groups of function fields of the Garcia-Stichtenoth towers. The main advantage of our construction is to allow more flexibility of localities.
\end{abstract}

\section{Introduction}
Due to applications in distributed storage and cloud storage systems, locally repairable codes have been extensively studied in recent years.
A block code is said to have locality $r$ if  every symbol of a given codeword can be recovered by accessing at most $r$ other coordinates of this codeword. In a distributed storage, when the data in one node (or disk) is erased, we want to use  data from a small set of other nodes to repair the data in the failure node. However, if some nodes in this small set are not available, we have to find an alternative set of nodes to repair the failure node. Thus, it is desirable to have multiple sets of nodes available to repair data in each node. In other words, for each node $A$, we need several disjoint sets of nodes such that data in any of these disjoint sets can be used to repair data in the node $A$. The number $t$ of the disjoint sets is called availability. The formal definition of \LRCs with availability $t$ is given in Definition \ref{def:1}.

Unlike in the classical coding case, only few papers study the asymptotic behavior of locally repairable codes with multiple recovering sets. The main purpose of this paper is to make use of algebraic geometry code to construct more asymptotically good locally repairable codes with multiple recovering sets.

\subsection{Locally repairable codes and availability}\label{subsec:2.1}
Let $t$ be a positive integer and $r_1,r_2,\cdots,r_t$ be $t$ positive integers.
Informally speaking, a block code is said with locality $(r_1,r_2,\cdots,r_t)$ if  every symbol of a given codeword can be recovered by accessing anyone of the $t$ disjoint recovering sets of at most $r_i$ other coordinates of this codeword for any $1\le i\le t$.
Let $q$ be a prime power and let $\F_q$ be the finite field with $q$ elements. The formal definition of a locally repairable code with multiple recovering sets can be given as follows.

\begin{defn}\label{def:1}
Let $C\subseteq \F_q^n$ be a $q$-ary block code of length $n$. For each $\Ga\in\F_q$ and $i\in \{1,2,\cdots, n\}$, define $C(i,\Ga):=\{\bc=(c_1,\dots, c_n)\in C\; : \; c_i=\Ga\}$. For a subset $I\subseteq \{1,2,\cdots, n\}\setminus \{i\}$, we denote by $C_{I}(i,\Ga)$ the projection of $C(i,\Ga)$ on $I$.
The code $C$ is called a locally repairable code with availability $t$ and locality $(r_1,r_2,\cdots, r_t)$ if, for every $i\in \{1,2,\cdots, n\}$, there exist disjoint subsets
$I_{i,j}\subseteq \{1,2,\cdots, n\}\setminus \{i\}$ for $1\le j\le t$ with $|I_{i,j}|\le r_j$ such that  $C_{I_{i,j}}(i,\Ga)$ and $C_{I_{i,j}}(i,\Gb)$ are disjoint for any $\Ga\neq \Gb\in \F_q$.
\end{defn}

A $q$-ary linear \LRC of length $n$, dimension $k$, minimum distance $d$ and locality $(r_1,r_2,\cdots,r_t)$ is denoted to be a $q$-ary $[n,k,d;(r_1,r_2,\cdots,r_t)]$-locally repairable code.
For $t=1$, the above definition gives the usual $[n,k,d; r]$-locally repairable codes.
The well-known Singleton-type bound for locally repairable codes with locality $r$  was given in \cite{GHSY12} by
\begin{equation}\label{Singletonbound}
d\le n-k-\left\lceil \frac{k}{r} \right\rceil+2.
\end{equation}
There are various upper bounds for the minimum distance of locally repairable codes with availability $t$ in the literature.
If $r_1=r_2=\cdots=r_t=r$, then Tamo and Barg provided an upper bound in \cite{TB14-C} for the minimum distance
\begin{equation}\label{TB}
d\le n-\sum_{i=0}^t \left\lfloor \frac{k-1}{r^i} \right\rfloor.
\end{equation}
In \cite{WZ14}, Wang {\it et al.} proved an upper bound on the minimum distance of $[n,k,d;(r,r,\cdots,r)]$-locally repairable codes as
\begin{equation}\label{WZ}
d\le n-k-\left\lceil \frac{(k-1)t+1}{(r-1)t+1} \right\rceil+2.
\end{equation}
The above bound \eqref{WZ} can be achieved when $n>k(rt+1)$.
In \cite{RPDV16}, the authors provided an upper bound on the minimum distance of $[n,k,d;(r,r,\cdots,r)]$-locally repairable codes as
\begin{equation}\label{RPDV}
d\le n-k-\left\lceil \frac{kt}{r} \right\rceil+t+1.
\end{equation}
It is easy to see that all these bounds \eqref{TB}, \eqref{WZ} and  \eqref{RPDV} are the generalization of Singleton-type bound \eqref{Singletonbound}.
As for possible different localities, Bhadane and Thangaraj provided an upper bound for $[n,k,d;(r_1,r_2,\cdots,r_t)]$-locally repairable codes with $r_1\le r_2\le \cdots \le r_t$  in \cite{BT17} as follows:
\begin{equation}\label{BT}
d\le n-k+1-\sum_{i=1}^t\left\lfloor \frac{k-1}{\prod_{j=t+1-i}^t r_j}\right\rfloor.
\end{equation}
In \cite{BMQ20}, the authors proved that the minimum distance of an $[n,k,d;(r_1,r_2,\cdots,r_t)]$-locally repairable code is upper bounded by 
\begin{equation}\label{BMQ}
d\le n-k-\left\lceil \frac{(k-1)t+1}{1+\sum_{i=1}^t r_i} \right\rceil+2.
\end{equation}

In  \cite{TB14}, Tamo and Barg provided a remarkable construction of optimal locally repairable codes attaining the Singleton-type bound  \eqref{Singletonbound} and locally repairable codes with two recovering sets via
good polynomials and subcodes of Reed-Solomon codes.
This method was generalized systematically to construct various locally repairable codes from algebraic curves in \cite{BTV17}.
In particular, Jin {\it et al.} constructed locally repairable codes with multiple recovering sets via automorphism groups of the rational function fields in \cite{JKZ20}.
By generalizing the construction in \cite{JKZ20}, Bartoli {\it et al.} constructed locally repairable codes with multiple recovery sets via automorphism groups of function fields with genus $g\ge 1$, including the Hermitian function fields, Giulleti-Korchmaros curves, the generalized Hermitian curve and the norm-trace curve \cite{BMQ20}.
Alternatively, fiber products of algebraic curves were used to construct locally repairable codes with multiple recovering sets in \cite{HMM18,CKMTW23}.

For availability $t=1$, the asymptotic Gilbert-Varshamov bound was proven for locally repairable codes with locality $r$ in \cite{TBF16}.
Such an asymptotic Gilbert-Varshamov bound of  locally repairable codes with one recovering set can be achieved by automorphism groups of function fields of the Garcia-Stichtenoth tower \cite{BTV17,LMX19-2} or adding columns to parity-check matrices of algebraic geometry codes constructed from the Garcia-Stichtenoth tower \cite{MX20,CKMTW23}.
However, it is difficult to derive the asymptotic  Gilbert-Varshamov bound for  locally repairable codes with multiple recovering sets as claimed in \cite{TBF16}.
Barg {\it et al.} \cite{BTV17} provided an asymptotic construction of locally repairable codes with two recovering sets via automorphism groups of function fields of the Garcia-Stichtenoth tower \cite{GS95}. In particular, let $q=\ell^2$ be a square of a prime power, let $r_1$ and $r_2$ be two positive integers such that $(r_1+1)|(\ell+1)$ and $(r_2+1)|\ell$, then there exists a family of $q$-ary $[n,k,d; (r_1,r_2)]$-linear locally repairable codes whose information rate $R$ and relative distance $\delta$ satisfy
\begin{equation}\label{tworecoverbound}
\delta+ \frac{(r_1+1)(r_2+1)}{r_1r_2} R  \ge \frac{\ell-2}{\ell-1}-\frac{r_1+r_2-2}{q-1}.
\end{equation}

\subsection{Our main results}\label{subsec:1.2}

To our best knowledge, the above construction in \cite{BTV17} is the only asymptotic construction of locally repairable codes with multiple recovering sets in the literature. Thus, it is desirable to design asymptotically good locally repairable codes with other parameter regimes of locality $(r_1,r_2)$.
 
In this paper, by investigating group structures of the automorphism groups of the function fields of the Garcia-Stichtenoth towers given in \cite{GS95,GS96}, we are able to construct asymptotically good locally repairable codes with more flexible localities. More precisely, we obtain $[n,k,d; (r_1,r_2)]$-locally repairable codes with locality $(r_1,r_2)$  given in the Table I whose information rate $R$ and relative distance $\delta$ satisfy
$$\delta+ \frac{(r_1+1)(r_2+1)}{r_1r_2-1} R  \ge \frac{\ell-2}{\ell-1}-\frac{r_1+r_2}{q-\ell}-\frac{1}{q-\ell}\cdot \frac{(r_1-r_2)^2}{r_1r_2-1}$$
or 
$$\delta+ \frac{(r_1+1)(r_2+1)}{r_1r_2-1} R  \ge \frac{\ell-2}{\ell-1}-\frac{r_1+r_2}{q-1}-\frac{1}{q-1}\cdot \frac{(r_1-r_2)^2}{r_1r_2-1}.$$
The Table I shows that our construction provides more flexibility for locally repairable codes with locality $(r_1,r_2)$. 

\begin{table}[]\label{tab:1}
	\setlength{\abovecaptionskip}{0pt}
	\setlength{\belowcaptionskip}{10pt}
	\caption{Localities of locally repairable codes with availability}
	\center
	\begin{tabular}{@{}|c|c|c|c|@{}}
		$r_1$ & $r_2$ & Restriction & Reference \\ 
		$(r_1+1)|(\ell+1)$ & $(r_2+1)|\ell$     &  $\gcd(r_1+1,r_2+1)=1$     &  \cite{BTV17}           \\ 
		$(r_1+1)|\ell$    & $(r_2+1)|(\ell-1)$     &   $(r_2+1)|(r_1,\ell-1) $     & Thm \ref{thm:3.3}                 \\
		$(r_1+1)|(\ell-1)$          & $(r_2+1)|(\ell-1)$  &  $\gcd(r_1+1,r_2+1)=1$  & Thm  \ref{thm:3.4}         \\ 
		$(r_1+1)|\ell$          & $(r_2+1)|\ell$   &  $(r_1+1)(r_2+1)\le \ell$    & Thm  \ref{thm:3.5}          \\ 

		$(r_1+1)|(\ell+1)$          & $(r_2+1)|(\ell+1)$ &    $\gcd(r_1+1,r_2+1)=1$   & Thm    \ref{thm:3.5}                  \\
	\end{tabular}
\end{table}

\subsection{Organization}
This paper is organized as follows. In Section 2, we present some preliminaries on function fields over finite fields, algebraic geometry codes, and the Garcia-Stichtenoth tower. In Section 3, we construct many families of asymptotically good locally repairable codes with multiple recovering sets via automorphism groups of function fields of the Garcia-Stichtenoth towers given in \cite{GS95,GS96}.

\section{Preliminaries}\label{sec:2}
In this section, we present some preliminaries on algebraic function fields over finite fields, algebraic geometry codes, and the asymptotically optimal Garcia-Stichtenoth tower of function fields given in \cite{GS96}.

\subsection{Algebraic function fields over finite fields}
Let $F/\F_q$ be an algebraic function field of one variable over the full constant field $\F_q$.
Let $\mathbb{P}_F$ denote the set of places of $F$ and let  $g(F)$ denote the genus of $F$.
Let $\nu_P$ be the normalized discrete valuation with respect to the place $P$.
The principal divisor of $z\in F^*$ is defined by $(z)=\sum_{P\in \mathbb{P}_F} \nu_P(z)P.$
For any divisor $G$ of $F$,
the Riemann-Roch space associated to $G$ is defined by $\mathcal{L}(G)=\{z\in F\setminus \{0\}: (z)\ge -G\}\cup \{0\}.$
It is a finite-dimensional vector space over $\F_q$ and its dimension is at least $\dim(G)\ge \deg(G)-g(F)+1$ from Riemann's theorem  \cite[Theorem 1.4.17]{St09}.
Let $E$ be a subfield of $F$ with the same full constant field $\F_q$ and let Diff$(F/E)$ be the different of $F/E$ which is an effective divisor.
The Hurwitz genus formula  \cite[Theorem 3.4.13]{St09} yields
$2g(F)-2=[F:E](2g(E)-2)+\deg \text{Diff}(F/E).$

Let $\Aut(F/\F_q)$ be the automorphism group of $F$ over $\F_q$, that is to say,
$\Aut(F/\F_q)=\{\sigma: F\rightarrow F  \mid   \sigma \text{ is an } \F_q \text{-automorphism of } F\}.$
Let $G$ be a subgroup of $\Aut(F/\F_q)$. The fixed subfield of $F$ with respect to $G$ is defined by
$F^{G}=\{z\in F: \Gs(z)=z \text{ for all } \Gs\in G\}.$
From Galois theory, $F/F^{G}$ is a Galois extension with $\Gal(F/F^{G})=G$.
For any automorphism $\Gs\in \Aut(F/\F_q)$, $\Gs(P)=\{\Gs(z):z\in P\}$ is a place of $F$ with the same degree.
The stabilizer $\{\Gs\in G: \Gs(P)=P\}$ is the decomposition group of $P$ in $F/F^{G}$.
For any place $P\in \mathbb{P}_F$, the place $P\cap F^{G}$ splits completely in $F$ if and only if $\Gs(P)$ are pairwise distinct for all automorphisms $\Gs\in G$.

\subsection{Algebraic geometry codes}
In the subsection, we introduce the construction of algebraic geometry codes whose subcodes can be employed to construct locally repairable codes \cite{TV91,BTV17}.
Let $F/\F_q$ be a function field of one variable over the full constant field $\F_q$. Let $\mP=\{P_1,\dots,P_n\}$ be a set of $n$ distinct rational places of $F$. For a divisor $G$ of $F$ with $0<\deg(G)<n$ and $\supp(G)\cap\mP=\emptyset$, the algebraic geometry code associated to $\mP$ and $G$ is defined to be
$C_{\mL}(\mP,G):=\{(f(P_1),f(P_2),\dots,f(P_n)): \; f\in\mL(G)\}.$
By \cite[Theorem 2.2.2]{St09}, $C_{\mL}(\mP,G)$ is an $[n,k,d]$-linear code with dimension $k=\dim(G)$ and minimum distance $d\ge n-\deg(G)$.
If $V$ is an $\F_q$-subspace of $\mL(G)$, we can define a subcode of $C_{\mL}(\mP,G)$ by
$C(\mP,V):=\{(f(P_1),f(P_2),\dots,f(P_n)):\; f\in  V\}.$
It is easy to see that the dimension of  $C(\mP,V)$ is the dimension of vector space $V$ over $\F_q$ and the minimum distance of  $C(\mP,V)$ is lower bounded by $n-\deg(G)$ as well.

\subsection{Garcia-Stichtenoth tower}\label{asymtower}
Let $\ell=p^w$ be a prime power and $q=\ell^2$ be a square.
In this subsection, we introduce the Garcia-Stichtenoth tower $\mathcal{T}=(T_1,T_2,T_3,\cdots)$ in \cite{GS96} which can be given by the rational function field $T_1=\F_q(y_1)$ and $T_m=T_{m-1}(y_m)$ with
\begin{equation}\label{optimaltower}
y_{m}^{\ell}+y_{m}=\frac{y_{m-1}^{\ell}}{y_{m-1}^{\ell -1}+1}
\end{equation}
recursively for $m\ge 2$.
The main properties of the Garcia-Stichtenoth tower can be summarized as follows from \cite{GS96,LMX19-2} .

\begin{prop}\label{prop:2.1}
\begin{itemize}
\item[(i)] The degree of extension $T_m/\F_q(y_i)$ is $[T_m:\F_q(y_i)]=\ell^{m-1}$ for any $i=1,2,\cdots,m.$
\item[(ii)]  The infinite place $\infty$ of $T_1$, i.e., the pole of $y_1$, is totally ramified in the extension $T_m/T_1$ for each $m\ge 2$.
Let $P_{\infty,m}$ be the unique place of $T_m$ lying over $\infty$. The rational place $P_{\infty,m}$ is a common pole of $y_1,y_2,\cdots,y_m$ and $\nu_{P_{\infty,m}}(y_m)=-1$.
\item[(iii)] For any $\alpha_1\in \F_q$ with $\Ga_1^\ell+\Ga_1\neq 0$, then the zero of $y_1-\Ga$ in $T_1$ splits completely in the extension $T_m/T_1$.
For any place $P$ of $T_m$ lying over $y_1-\Ga_1$, there exist $\Ga_i\in \F_q^*$ with $\Ga_i^\ell+\Ga_i=\Ga_{i-1}^{\ell}/(\Ga_{i-1}^{\ell-1}+1)$ for all $2\le i\le m$
such that $P$ is the unique common zero of $y_1-\Ga_1,\cdots,y_m-\Ga_m$ in $T_m$. Hence, $P$ can be identified with the $m$-tuple $(\Ga_1,\Ga_2,\cdots,\Ga_m)$ and $y_i(P)=\alpha_i$ for $1\le i\le m$.
\item[(iv)]  The genus of $T_m$ is given by $$g(T_m)=\begin{cases}
(\ell^{\frac{m}{2}}-1)^2, & \text{ if } m \equiv 0(\mbox{mod } 2),\\
(\ell^{\frac{m+1}{2}}-1)(\ell^{\frac{m-1}{2}}-1), & \text{ if } m \equiv 1(\mbox{mod } 2).
\end{cases}$$
\item[(v)]  The Garcia-Stichtenoth tower  $\mathcal{T}=(T_1,T_2,T_3,\cdots)$ is asymptotically optimal, i.e.,
$$\lim_{m\rightarrow \infty} \frac{N(T_m)}{g(T_m)}=\ell-1.$$
\end{itemize}
\end{prop}

Let $\mA$ be the set of automorphisms $\Gs$ of $T_m$ over $\F_q$ with the following form
$$\begin{cases}
\Gs(y_i) = cy_i  \text{ for } i=1,2,\cdots, m-1,\\
\Gs(y_m) = cy_m+a,
\end{cases} $$
where $c\in \F_{\ell}^*$ and $a^\ell+a=0$.
Let $u$ be a divisor of $\ell-1$ and $H$ be the cyclic  subgroup with order $u$ of the multiplicative group $\F_\ell^*$.
Let $v$ be an integer with $0\le v\le w$ and $u|(p^v-1)$.
Let  $h=\min \{ t\in \mathbb{Z}_+: u|(p^t-1)\}$. Then $\F_p(H)=\F_{p^h}$ and it is a subfield of $\F_\ell$ and $\F_{p^v}$.
Let $W$ be an  $\F_{p^h}$-subspace of $\{a\in \F_q: a^\ell+a=0\}$ with dimension $v/h$.
Let $H_1=\{\sigma\in \mA: \sigma(y_i)=y_i\text{ for } 1\le i\le m-1, \Gs(y_m)=y_m+a\text{ for }  a\in W\}$ and $H_2=\{\sigma\in \mA: \sigma(y_i)=cy_i, c\in H, 1\le i\le m\}$.
From \cite[Proposiiton III.2]{LMX19-2}, $G=H_1H_2$ is a subgroup of $\mA$ with order $up^v$.

\begin{prop}\label{prop:2.2}
Let $\ell=p^w$ be a prime power.
Let $v$ be an integer with $0\le v\le w$ and let $u$ be a positive integer satisfying $u|\gcd(p^v-1, \ell-1)$.
Let $H_1$ and $H_2$ be subgroups of $\mA$ defined as above and $G=H_1H_2$.
Then $G$ is a semi-direct product of $H_1$ and $H_2$.
\end{prop}
\begin{proof}
It is easy to see that $H_1\cap H_2=\{\text{id}\}$. Let $\sigma\in H_1$ be an automorphism of $T_m$ given by
$$\begin{cases}
\Gs(y_i) = y_i  \text{ for } i=1,\cdots, m-1,\\
\Gs(y_m)= y_m+a_1,
\end{cases} $$
with  $a_1\in W$. Let $\tau\in G$ be an automorphism of $T_m$ with
$$\begin{cases}
\tau(y_i)= cy_i  \text{ for } i=1,\cdots, m-1,\\
\tau(y_m)= cy_m+a_2,
\end{cases} $$
where  $c\in H$ and $a_2\in W$. It is easy to verify that $\tau^{-1}\Gs\tau\in H_1$, i.e.,
$$\begin{cases}
\tau^{-1}\Gs\tau(y_i) =y_i  \text{ for } i=1,\cdots, m-1,\\
\tau^{-1}\Gs\tau(y_m) = y_m+ca_1.
\end{cases} $$
Hence, $G$ is a semi-direct product of $H_1$ and $H_2$.
\end{proof}

\section{Asymptotically good locally repairable codes with availability}
In this section, we first introduce a general construction of locally repairable codes with two disjoint recovering sets from automorphism groups of function fields and then employ this method to construct asymptotically good families of locally repairable codes with two recovering sets from the Garcia-Stichtenoth towers.

\subsection{Group-theoretic construction of locally repairable codes with availability}
In this subsection, we present a general construction of locally repairable codes with two disjoint recovery sets from automorphism groups of function fields similarly as \cite{BTV17,JKZ20,BMQ20}. This technique has been used to construct optimal locally repairable codes via automorphism groups of algebraic curves and algebraic surfaces \cite{JMX20,LMX19,SVV21,MX23}.

Let $F/\F_q$ be an algebraic function field of one variable over the full constant field $\F_q$. Let $\Aut(F/\F_q)$ be the automorphism group of $F$ over $\F_q$.
For $i=1,2$, let $H_i$ be a subgroup of $\Aut(F/\F_q)$ of order $r_i+1$ and let $F_i$ be the fixed subfield of $F$ with respect to $H_i$.
Assume that $\Aut(F/\F_q)$ has a subgroup $G$ which is a semi-direct product of $H_1$ and $H_2$, i.e., $G=H_1\rtimes H_2$.
Denote by $E$ the fixed subfield $F^{G}$. From Galois theory, $F/E$ is a Galois extension with Galois group $\Gal(F/E)=G.$

Assume that there exist $s$ rational places $Q_1,Q_2,\cdots, Q_s$ of $E$ which are all splitting completely in $F/E$.
Let $P_{i,1}, P_{i,2},\cdots, P_{i,(r_1+1)(r_2+1)}$ be the $(r_1+1)(r_2+1)$ rational places of $F$ lying over $Q_i$ for each $1\le i \le s$.
Put $\mP=\{P_{i,j}: 1\le i \le s, 1\le j\le (r_1+1)(r_2+1)\}$ be the set of evaluation places. Hence, the cardinality of $\mP$ is $n=s(r_1+1)(r_2+1)$.

Suppose that there exist $w_i\in F$ for $i=1,2$ satisfying the following conditions
\begin{itemize}
\item[(1)] $1,w_{i},w_{i}^{2},\cdots, w_{i}^{r_i-1}$ are linearly independent over $F_i$;
\item[(2)] $w_{i}$ takes pairwise distinct evaluations on the rational places of $\{\sigma(P): \sigma\in H_i\}$ for any $P\in \mP$.
\end{itemize}
For a positive integer $1\le d\le n$, choose an effective divisor $D$ of $F$ with degree $n-d$ such that $\supp(D)\cap \supp(\mP)=\emptyset$ and $(r_i-1)(w_i)_\infty\le D$ for $i=1,2$.
Let $D_i$ be the largest possible effective divisor of $F_{i}$ such that $(r_i-1)(w_i)_\infty+Conorm_{F/F_{i}}(D_i)\le D$ for $i=1,2$.
Let $\{f_{i,j}: 1\le j\le \dim(D_i)\}$ be a basis of $L(D_i)$. Let $V_i$ be  subspaces of the Riemann-Roch space $\mL(D)$ defined by
$$V_{i}=\left\{\sum_{l=0}^{r_i-1}\left(\sum_{j=1}^{\dim(D_i)} c_{i,j,l} f_{i,j}\right)w_i^l: \forall c_{i,j,l}\in \F_q\right\}.$$
Let $V=V_1\cap V_2$ be a non-trivial subspace of  $\mL(D)$.
Let $ev_{\mP}$ be the evaluation map
$$ev_{\mP}:V\rightarrow \F_q^n; f\mapsto (f(P_1),f(P_2),\cdots,f(P_n)).$$
We shall show that the image $ev_{\mP}(V)$ of the evaluation map
is an $[n,\dim(V),\ge d; (r_1,r_2)]$-locally repairable codes with two disjoint recovering sets.

\begin{prop}\label{prop:3.1}
Let $\mP$ and $V$ be defined as above. Then the image of the evaluation map
$$ev_{\mP}:V\rightarrow \F_q^n; f\mapsto (f(P_1),f(P_2),\cdots,f(P_n))$$
is an $[n,\dim(V),\ge d; (r_1,r_2)]$-locally repairable codes with two disjoint recovering sets.
\end{prop}

\begin{proof}
Firstly, $ev_{\mP}(V)$ is a subcode of the algebraic geometry code $C_{\mL}(\mP, D)$. Hence, the minimum distance of $ev_{\mP}(V)$ is at least $n-\deg(D)=d$.
By \cite [Theorem 3.1]{LMX19-2}, the image of the evaluation map
$$ev_{\mP,i}:V_i\rightarrow \F_q^n; f\mapsto (f(P_1),f(P_2),\cdots,f(P_n))$$
are $[n,\dim(V_i),\ge d; r_i]$-locally repairable codes for $i=1,2$.
In fact, the locality property follows from Lagrange interpolation, since the element $w_i$ takes pairwise distinct values on the set of $r_i$ rational places
$\{\Gs(P): \text{id}\neq \Gs\in H_i\}$ for any rational place $P\in \mP$.
For any $\sigma\in H_1$ and $\tau\in H_2$, if $\sigma(P)=\tau(P)$, then $\sigma^{-1}\tau(P)=P$. Since $P\cap E$ splits completely in $F$, we must have $\sigma^{-1}\tau=\text{id}$, i.e., $\sigma=\tau\in H_1\cap H_2=\{\text{id}\}$.
Hence, we have $\{\sigma(P): \text{id}\neq \sigma\in H_1\}\cap \{\tau(P):\text{id}\neq \tau\in H_2\}=\emptyset$ for any place $P\in \mP$, i.e., two recovering sets of $P$ are disjoint. 
Hence, the image of the evaluation map $ev_{\mP}$ is an $[n,\dim(V),\ge d; (r_1,r_2)]$-locally repairable codes with two disjoint recovering sets.
\end{proof}

\subsection{Asymptotic locally repairable codes with availability: I}
In this subsection, we will construct locally repairable codes with two disjoint recovering sets via automorphism groups of function fields of
Garcia-Stichtenoth tower $\mathcal{T}=(T_1,T_2,T_3,\cdots)$ in Section \ref{asymtower}.

Let $\mA$ be the set of all automorphisms $\Gs$ of $T_m$ over $\F_q$ with the form
$$\begin{cases}
\Gs(y_i) = cy_i  \text{ for } i=1,\cdots, m-1,\\
\Gs(y_m) = cy_m+a,
\end{cases} $$
where $c\in \F_{\ell}^*$ and $a^\ell+a=0$.
Let $H_1$ and $H_2$ be subgroups of $\mA$ with order $r_1+1=p^v$ and $ r_2+1=u$ such that $G=H_1\rtimes H_2$ is a subgroup of $\mA$ with order $(r_1+1)(r_2+1)$ in Proposition \ref{prop:2.2}.
Let $T_{m,i}$ be the fixed subfields of $T_m$ with respect to $H_i$.
From the Hurwitz genus formula, we have
$2g(T_m)-2\ge (r_i+1)[2g(T_{m,i})-2]$, that is to say
$$g(T_{m,i})-1\le \frac{g(T_m)-1}{r_i+1}.  $$
Let $\mP_m$ be a subset of $\mathbb{P}_{T_m}$ which is defined by
$\mP_m=\big{\{}(\Ga_1,\Ga_2,\cdots,\Ga_m):\Ga_1^{\ell}+\Ga_1\neq 0, \Ga_i^\ell+\Ga_i=\frac{\Ga_{i-1}^{\ell}}{\Ga_{i-1}^{\ell-1}+1} \text{ for } 2\le i\le m\big{\}}.$

\begin{lemma}\label{lem:3.2}
\begin{itemize}
\item[(1)] $1,y_{m},y_{m}^{2},\cdots, y_{m}^{r_i-1}$ are linearly independent over $T_{m,i}$;
\item[(2)] $y_{m}$ takes pairwise distinct evaluations on the rational places of $\{\sigma(P): \sigma\in H_i\}$ for any $P\in \mP_m$.
\item[(3)] $P\cap T_m^{G}$ splits completely in the extension $T_m/T_m^{G}$ for any place $P\in \mP_m$.
\end{itemize}
\end{lemma}
\begin{proof}
Please refer to the proof of \cite[Theorem III.3]{LMX19-2}.
\end{proof}

Now we can construct asymptotically good locally repairable codes with two disjoint recovering sets from subgroups of $\mA$.

\begin{theorem}\label{thm:3.3}
Let $q=\ell^2$ with $\ell=p^w$. For any integer $v$ with $1\le v\le w$ and any positive integer $u$ satisfying
$u|\gcd(p^v-1,\ell-1)$.
Let $r_1$ and $r_2$ be two positive integers such that $r_1+1=p^v$ and $r_2+1=u$. For any $0\le \delta\le1-2/(\ell-1)$, there exists a family of $q$-ary $[n_m,k_m,\ge d_m; (r_1,r_2)]$-linear locally repairable codes whose information rate $R$ and relative distance $\delta$ satisfy
$$\delta+ \frac{(r_1+1)(r_2+1)}{r_1r_2-1} R  \ge \frac{\ell-2}{\ell-1}-\frac{r_1+r_2}{q-\ell}-\frac{1}{q-\ell}\cdot \frac{(r_1-r_2)^2}{r_1r_2-1}.$$
\end{theorem}

\begin{proof}
Choose an effective divisor $D$ of $T_m$ with degree $n_m-d_m$ such that $\supp(D)\cap \supp(\mP)=\emptyset$ and $(r_i-1)\cdot (y_m)_\infty\le D$ for $i=1,2$. 
Let $D_i$ be the effective divisor of $T_{m,i}$ with the largest possible degree such that $(r_i-1)\cdot (y_m)_\infty+Conorm_{T_m/T_{m,i}}(D_i)\le D$ for $i=1,2$.
Let $\{f_{i,j}: 1\le j\le \dim(D_i)\}$ be a basis of $L(D_i)$. Let
$$V_{m,i}=\left\{\sum_{l=0}^{r_i-1}\left(\sum_{j=1}^{\dim(D_i)} c_{i,j,l} f_{i,j}\right)y_m^l: \forall c_{i,j,l}\in \F_q\right\}$$
be two subspaces of $\mL(D)$. By Lemma \ref{lem:3.2} and \cite[Theorem III.3]{LMX19-2}, the image of the evaluation map $ev_{\mP_m,i}: V_{m,i}\rightarrow \F_q^n$ are $[n_m=(q-\ell)\ell^{m-1},k_{m,i},\ge d_m; r_i]$-locally repairable codes for $i=1,2$.
If the designed minimum distance is at least $n_m-(r_i-1)\ell^{m-1}-(r_i+1)\deg(D_i)\ge d_m$, then we have
$$\deg(D_i)+1\ge \frac{n_m-d_m-(r_i-1)\ell^{m-1}}{r_i+1}$$
and
$k_{m,i}=r_i\dim(D_i)\ge r_i[\deg(D_i)+1-g(T_{m,i})].$
Let $V_m=V_{m,1}\cap V_{m,2}$. If $0\le \delta\le1-2/(\ell-1)$, then we can assume that $\deg(D)\ge 2g(T_m)-1$. Hence, we have $\dim(D)=\deg(D)-g(T_m)+1=n_m-d_m-g(T_m)+1$.
From the direct computation, the dimension of $V_m$ is determined by
\begin{align*}
k_m&\ge k_{m,1}+k_{m,2}-\dim(D)\\
&\ge  \frac{r_1r_2-1}{(r_1+1)(r_2+1)} (n_m-d_m-g(T_m)+1)\\
&\ \ \ -\frac{n_m}{q-\ell}\left[\frac{r_1(r_1-1)}{r_1+1}+\frac{r_2(r_2-1)}{r_2+1}\right]-r_1-r_2.
\end{align*}
By dividing $n_m$ at both sides of the above inequality, we have
\begin{align*}
\frac{k_m}{n_m}\ge  & \frac{r_1r_2-1}{(r_1+1)(r_2+1)}\left(1-\frac{d_m}{n_m}-\frac{g(T_m)-1}{n_m}\right)&\\
& -\frac{1}{q-\ell}\left[\frac{r_1(r_1-1)}{r_1+1}+\frac{r_2(r_2-1)}{r_2+1}\right]-\frac{r_1+r_2}{n_m}.
\end{align*}
Let $m$ approach to $\infty$ and take limits. Thus, we have
$$\delta+ \frac{(r_1+1)(r_2+1)}{r_1r_2-1} R  \ge \frac{\ell-2}{\ell-1}-\frac{r_1+r_2}{q-\ell}-\frac{1}{q-\ell}\cdot \frac{(r_1-r_2)^2}{r_1r_2-1}.$$
This completes the proof.
\end{proof}

Let $H_1$ and $H_2$ be two subgroups of $\mA$ consisting of automorphisms with $c=0$ or let $H_1$ and $H_2$ be two subgroups of $\mA$ with $a=1$. By Proposition \ref{prop:3.1} and Theorem \ref{thm:3.3}, we have the following results.
\begin{theorem}\label{thm:3.4}
Let $q=\ell^2$ with $\ell=p^w$. Let $r_1$ and $r_2$ be two positive integers satisfying one of the conditions:
\begin{itemize}
\item[(1)] $(r_1+1)|(\ell-1)$, $(r_2+1)|(\ell-1)$ and $\gcd(r_1+1,r_2+1)=1$;
\item[(2)] $(r_1+1)|\ell$, $(r_2+1)|\ell$ and $(r_1+1)\cdot (r_2+1)\le \ell$.
\end{itemize}
For any $0\le \delta\le1-2/(\ell-1)$, there exists a family of $q$-ary $[n_m,k_m,\ge d_m; (r_1,r_2)]$-linear locally repairable codes whose information rate $R$ and relative distance $\delta$ satisfy
$$\delta+ \frac{(r_1+1)(r_2+1)}{r_1r_2-1} R  \ge \frac{\ell-2}{\ell-1}-\frac{r_1+r_2}{q-\ell}-\frac{1}{q-\ell}\cdot \frac{(r_1-r_2)^2}{r_1r_2-1}.$$
\end{theorem}

\subsection{Asymptotic locally repairable codes with availability: II}
In this subsection, let us consider function fields of the asymptotically optimal Garcia-Stichtenoth tower $\mathcal{T}=(T_1,T_2,T_3,\cdots)$ which are given by the rational function field $T_1=\F_q(x_1)$ and $T_m=T_{m-1}(z_m)$ defined by
$z_{m}^{\ell}+z_{m}=x_{m-1}^{\ell+1}$
with $x_m=z_m/x_{m-1}$ recursively for $m\ge 2$ in \cite{GS95}.
The zero of $x_1-\alpha_1$ in $T_1$ splits completely in $T_m/T_1$ for any $m\ge 2$ and $\alpha_1\in \F_q\setminus \{0\}$ from \cite{GS95}.
Let  $\mP_m$ be the set of all rational places of $T_m$ lying above the zero of $x_1-\alpha_1$ with $\alpha_1\neq 0$.
Hence, the cardinality of $\mP_m$ is $(q-1)\ell^{m-1}$. 

\begin{theorem}\label{thm:3.5}
Let $q=\ell^2$ with $\ell=p^w$. Let $r_1$ and $r_2$ be two positive integers satisfying one of the conditions:  
\begin{itemize}
\item[(1)] $(r_1+1)|(\ell+1)$, $(r_2+1)|(\ell+1)$ and $\gcd(r_1+1,r_2+1)=1$;
\item[(2)] $(r_1+1)|\ell$, $(r_2+1)|\ell$ and $(r_1+1)\cdot (r_2+1)\le \ell$. 
\end{itemize}
For any $0\le \delta\le1-2/(\ell-1)$, there exists a family of $q$-ary $[n_m,k_m,\ge d_m; (r_1,r_2)]$-linear locally repairable codes whose information rate $R$ and relative distance $\delta$ satisfy
$$\delta+ \frac{(r_1+1)(r_2+1)}{r_1r_2-1} R  \ge \frac{\ell-2}{\ell-1}-\frac{r_1+r_2}{q-1}-\frac{1}{q-1}\cdot \frac{(r_1-r_2)^2}{r_1r_2-1}.$$
\end{theorem}
\begin{proof}
\begin{itemize}
\item[(1)] Let $H_1$ and $H_2$ be two subgroups of $\mB$ consisting of automorphisms with 
$$\begin{cases}
\Gs(x_1) = ax_1,\\
\Gs(z_i) = z_i, \text{ for } i=2,3,\cdots, m,\\
\end{cases} $$
where $a^{\ell+1}=1$. 
For any $P\in \mP_m$, it comes from the common zero of $x_1-\alpha_1,z_2-\alpha_2,\cdots, z_m-\alpha_m$ and can be denoted by $(\Ga_1,\Ga_2,\cdots,\Ga_m)$. 
It is easy to verify that $x_1(\Gs^{-1}(P))=\Gs(x_1)(P)=ax_1(P)=a\Ga_1$.
Hence, $\sigma^{-1}(P)=(a\Ga_1,\Ga_2,\cdots,\Ga_{m-1}, \Ga_{m})$ are pairwise distinct  for all $\sigma\in \mB$ and $x_1$ takes pairwise distinct  evaluations on rational places of $\{\sigma^{-1}(P): \sigma\in \mB\}$. 
\item[(2)]   Let $H_1$ and $H_2$ be two subgroups of $\mC$ consisting of automorphisms with 
$$\begin{cases}
\Gs(x_1) = x_1, \quad 
\Gs(z_i) = z_i \text{ for } i=2,\cdots, m-1,\\
\Gs(z_m) = z_m+c,
\end{cases} $$
where $c^\ell+c=0$. 
It is easy to verify that $z_m(\Gs^{-1}(P))=\Gs(z_m)(P)=(z_m+c)(P)=\Ga_m+c$.
Hence, $\sigma^{-1}(P)=(\Ga_1,\Ga_2,\cdots,\Ga_{m-1}, \Ga_{m}+c)$ are pairwise distinct  for all $\sigma\in  \mC$ and $z_m$ takes pairwise distinct evaluations on rational places of $\{\sigma^{-1}(P): \sigma\in \mC\}$. 
\end{itemize}
By  Proposition  \ref{prop:3.1}, this theorem can be proved similarly as Theorem \ref{thm:3.3}. We omit the details.
\end{proof}

\section{Conclusion}
In this manuscript, we provide more asymptotic construction of locally repairable codes with multiple recovering sets by investigating the group structures of the automorphism groups of function fields of the Garcia-Stichtenoth towers given in \cite{GS95,GS96}.

\end{document}